\documentclass[submission,copyright,creativecommons]{eptcs}
\providecommand{\event}{AiML 2026} 

\usepackage{aiml26}

\usepackage{iftex}
\usepackage{stmaryrd}

\ifpdf
  \usepackage{underscore}         
  \usepackage[T1]{fontenc}        
\else
  \usepackage{breakurl}           
\fi

\usepackage{graphicx}
\usepackage{tikz}
\usepackage{subfigure}
\usepackage{makeidx}
\usepackage{pgf}
\usepackage{gastex}
\usepackage{adjustbox}

\usetikzlibrary{shapes.misc}

\tikzset{cross/.style={cross out, draw=black, minimum size=2*(#1-\pgflinewidth), inner sep=0pt, outer sep=0pt},
cross/.default={3pt}}

\usetikzlibrary{arrows,
	automata,
	calc,
	decorations,
	decorations.pathreplacing,
	positioning,
	shapes}

\usepackage{amsthm}
\usepackage{amsmath}
\usepackage{amsfonts}
\usepackage{amssymb}
\usepackage{thmtools} 
\usepackage{thm-restate}

\theoremstyle{plain}
\newtheorem{theorem}{Theorem}[section]

\newtheorem{proposition}[theorem]{Proposition}
\newtheorem{lemma}[theorem]{Lemma}
\newtheorem{corollary}[theorem]{Corollary}
\theoremstyle{definition}

\newtheorem{definition}[theorem]{Definition}

\newtheorem{remark}[theorem]{Remark}
\newtheorem{example}[theorem]{Example}

\newif\ifsolution

\newcommand{\AP}{\mathsf{AP}}

\newcommand{\NF}{\mathfrak{NF}}
\newcommand{\KF}{\mathfrak{KF}}

\newcommand{\Log}{\mathit{Log}}

\newcommand{\TNL}{\mathsf{T}_3\mathsf{m}}
\newcommand{\DNL}{\mathsf{D}_3\mathsf{m}}

\newcommand{\Nr}{{N_\omega^{r}}}
\newcommand{\Ni}{{N_\omega^{i}}}
\newcommand{\Tr}{T_\omega^{r}}
\newcommand{\Tsc}{T_{\omega,\omega,\omega}^{r,\mathsf{m}}}
\newcommand{\Ti}{T_\omega^{i}}
\newcommand{\Tim}{T_{\omega,\omega,\omega}^{i,\mathsf{m}}}

\newcommand{\forget}{\mathit{forget}}

\title{On Modal Logics of Full Products of Neighborhood Frames}

\author{
Rajab Aghamov$^{1}$,
Andrey Kudinov$^{2,3}$,
Maik Thanh Nguyen$^{1}$,
Jakob Piribauer$^{1}$\\[1ex]
$^{1}$Technische Universit\"{a}t Dresden, Dresden, Germany\\
$^{2}$Higher School of Modern Mathematics, MIPT, Moscow, Russia\\
$^{3}$HSE University, Moscow, Russia
}

\newcommand{\titlerunning}{On Modal Logics of Full Products of Neighborhood Frames}
\newcommand{\authorrunning}{R. Aghamov, A. Kudinov, M. T. Nguyen, J. Piribauer}
\newcommand{\copyrightholders}{Aghamov, Kudinov, Nguyen \& Piribauer}
\definecolor{wred}{RGB}{150,30,30}   
\definecolor{mblue}{rgb}{0.0, 0.53, 0.74}
\definecolor{sgreen}{rgb}{0.53, 0.66, 0.42}

\begin{document}

\maketitle

\begin{abstract}
\textbf{Abstract.}  On the product of two neighborhood frames, three natural neighborhood functions can be defined:
the horizontal one assigning to a point $(x,y)$ the set of all supersets of sets $U\times\{y\}$, where $U$ is a neighborhood of $x$;
the vertical analog; and
the product neighborhood function assigning as neighborhoods  all supersets of sets $U\times V$, for neighborhoods $U$ of $x$ and $V$ of $y$.

We define the tri-modal logics $\mathsf{T}\times_n^+ \mathsf{T}$ and $\mathsf{D}\times_n^+ \mathsf{D}$ of classes of full products equipped with all three neighborhood functions of neighborhood frames validating the logic $\mathsf{T}$ or $\mathsf{D}$; thereby extending known product results for $\mathsf{S4}$ and $\mathsf{D4}$ to weaker systems.
Two interaction principles arise: $(\mathsf{sub})=\Box p \rightarrow \Box_1 p \land \Box_2 p$ and $(\mathsf{mix})=\Box p \rightarrow \Box_1\Box_2 p \land \Box_2\Box_1 p$, where $\Box$ for the product neighborhood function and $\Box_1,\Box_2$ the horizontal and vertical ones.
Namely, we show that
$\mathsf{T}\times_n^+ \mathsf{T} = \mathsf{T}\otimes \mathsf{T}\otimes \mathsf{T}  + (\mathsf{mix})$
and
$\mathsf{D}\times_n^+ \mathsf{D} = \mathsf{D}\otimes \mathsf{D}\otimes \mathsf{D} + (\mathsf{mix})$,
where $\otimes$ denotes fusion.
Notably, $(\mathsf{sub})$ and $(\mathsf{mix})$ are equivalent over $\mathsf{S4}\otimes\mathsf{S4}\otimes\mathsf{S4}$ and thus  $\mathsf{S4}\otimes\mathsf{S4}\otimes\mathsf{S4}+(\mathsf{mix})$ axiomatizes the logic of full products of topological spaces.

\end{abstract}

\paragraph*{Acknowledgments.}
The authors were supported by DFG grant 389792660
as part of TRR 248 (see https://perspicuous-computing.science) and by the BMBF (Federal Ministry of Education and
Research) in DAAD project 57616814 (SECAI, School of
Embedded and Composite AI) as part of the program Konrad Zuse Schools of Excellence in Artificial Intelligence.

\section{Introduction}
Neighborhood semantics generalizes relational Kripke semantics and was introduced in the 1970s by Scott \cite{Scott} and Montague \cite{mont} to study weak, in particular non-normal, modal logics.
A neighborhood model $M=(X,\tau,V)$ consists of a domain $X$, a function $\tau\colon X\to 2^{2^X}$ assigning to each $x\in X$ a set of neighborhoods $\tau(x)\subseteq 2^X$, and a valuation $V\colon \mathsf{AP}\to 2^X$. 
In this work, each $\tau(x)$ is required to be a filter on $X$, i.e., a non-empty family of subsets of $X$ closed under supersets and finite intersections, which ensures normality of the resulting logics \cite{Pacuit}. 
The modal clause is given by $M,x\vDash\Box\varphi$ iff $\llbracket\varphi\rrbracket^M\in\tau(x)$, where $\llbracket\varphi\rrbracket^M=\{y\in X\mid M,y\vDash\varphi\}$.


Neighborhood semantics is conceptually similar to topological semantics. 
There are two standard interpretations: the diamond-as-closure (equivalently, box-as-interior) semantics and the diamond-as-derivation semantics; both can be viewed as special cases of neighborhood semantics. 
Given a topological space $(Y,\sigma)$, the closure interpretation assigns to each $y\in Y$ the family of all sets containing some open neighborhood of $y$. 
This approach was initiated by McKinsey and Tarski \cite{mctarski} in 1944. 
In the derivational interpretation, the neighborhood function induced by a topological space assigns to each $y\in Y$ the set of all supersets of  $U\setminus\{y\}$, where $U$ is an open neighborhood of $y$; this semantics was later studied by Esakia \cite{Esakia} and Shehtman \cite{vbsder}. 
Unless stated otherwise, “topological semantics” refers to the closure interpretation in this paper.


In topological semantics, the axioms 
$\mathsf{(T)}=\Box p\to p$ and 
$\mathsf{(4)}=\Box p\to\Box\Box p$ 
are valid. Hence, the logic of any class of topological spaces subsumes $\mathsf{S4}$, the smallest normal modal logic containing $\mathsf{(T)}$ and $\mathsf{(4)}$, also denoted $\mathsf{K}+\mathsf{(T)}+\mathsf{(4)}$. 
In contrast, in neighborhood semantics both axioms can be falsified, making this framework suitable for the study of logics below $\mathsf{S4}$. In this paper, in particular the following logics are of interest, where $\mathsf{(D)}=\Box p\to\Diamond p$ is the seriality axiom:
\[
\mathsf{D}=\mathsf{K}+\mathsf{(D)}, \qquad 
\mathsf{D4}=\mathsf{K}+\mathsf{(D)}+\mathsf{(4)}, \qquad 
\mathsf{T}=\mathsf{K}+\mathsf{(T)}.
\]
Since $\mathsf{(D)}$ is derivable in $\mathsf{T}$, all these logics contain $\mathsf{(D)}$; we therefore refer to them as \emph{serial} logics.


\vspace{3pt}
\noindent\textbf{Products of frames.}
When studying multi-modal systems whose modalities refer to different dimensions, such as space and time, or to different agents, it is often necessary to combine frames for uni-modal logics into a single multi-modal structure. 
A natural way to do so is via a product construction.
Also from the perspective of epistemic logic, for instance, combining models for several agents may be achieved by taking the product of their underlying frames.

For Kripke frames, the resulting product is given by the Cartesian product of the underlying frames equipped with two accessibility relations. Such relational products were first studied by Segerberg~\cite{Segerberg1973} and Shehtman~\cite{vbs1978}. They validate interaction principles such as commutation $\mathsf{(com)}$ 
and convergence known as the Church-Rosser principle $\mathsf{(chr)}$
\[
\mathsf{(com)}=\Box_1\Box_2 p \leftrightarrow \Box_2\Box_1 p
\qquad\text{and}\qquad
\mathsf{(chr)}=\Diamond_1\Box_2 p \to \Box_2\Diamond_1 p.
\]

In application domains such as epistemic logic, however, these interaction principles are often not desirable.
\emph{Topological semantics} provides a way to weaken this built-in interaction: Van Benthem et al.~\cite{vanbenguram} introduce on $X\times Y$ the \emph{horizontal} and \emph{vertical} topologies obtained from the component topologies by fixing one coordinate.
These topologies provide a direct analogue of the Kripke product. In neighborhood semantics,  analogous products of neighborhood frames resulting in a frame with a \emph{vertical} and \emph{horizontal} neighborhood function were introduced  by Sano in \cite{sano}.
The interaction axioms $\mathsf{(com)}$ and $\mathsf{(chr)}$ may fail in products of topological spaces and consequently in the more general products of neighborhood frames. So, topological products and neighborhood products allow one to combine two information dimensions without enforcing interaction principles between the corresponding modalities a priori. 

\begin{example}[\cite{vanbenguram}]
We illustrate how axioms like $\mathsf{(com)}$ can be falsified in topological semantics. As topological spaces can be viewed as neighborhood frames 
and the product construction coincides in both views, this also shows how $\mathsf{(com)}$ can be falsified in neighborhood semantics.

\begin{figure}
	\centering
\begin{tikzpicture}[scale=2]

	\draw[semithick] (-1.3,0) -- (1.3,0);
	\draw[semithick] (0,-0.7) -- (0,0.7);
	\draw[dashed] (0.02,-0.7) -- (0.02,0.7);

	\draw[dashed] (0,0) -- (1,0.4);
	\draw[dashed] (0,0) -- (1,-0.4);
	
	\draw[dashed] (0,0) -- (-1,0.4);
	\draw[dashed] (0,0) -- (-1,-0.4);
	
	\foreach \x in {0.1,0.2,...,0.9}{
		\draw[dotted] (\x,{-0.4*\x}) -- (\x,{0.4*\x});
	}
	\foreach \x in {-0.1,-0.2,...,-0.9}{
		\draw[dotted] (\x,{0.4*abs(\x)}) -- (\x,{0.4*-abs(\x)});
	}

    \node[text=wred] at (0.17,-0.2) {\scriptsize $(0,0)$};
    \fill[wred] (0,0) circle (0.6pt);

	\node at (0.6,0.4) {\scriptsize $p$};

\end{tikzpicture}
\caption{Regions where proposition $p$ is true.}
\label{fig:counter}
\end{figure}
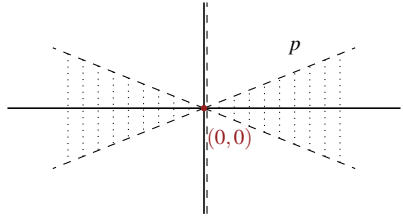

Consider the product of $\mathbb{R}$ with the standard topology $\sigma$ with itself, resulting in
the product $X = (\mathbb{R}\times\mathbb{R},\sigma_1,\sigma_2)$ where the vertical topology $\sigma_1$ contains all sets of the form $U\times \{y\}$ for an open set $U\in \sigma$
and $y\in \mathbb{R}$ and the horizontal topology $\sigma_2$ similarly contains all sets of the form $\{x\}\times U$.
Now, consider the valuation 
$
V(p) = \{(x,y)\in \mathbb{R}\times\mathbb{R} \mid |y|<|x|\} \cup \{0\}\times \mathbb{R}$.
The region where $p$ holds is depicted in Figure~\ref{fig:counter}.
Then, $X,(0,0)\vDash \Box_1\Box_2 p$: For any point $(x,0)$ in the set $ (-1,1)\times\{0\}\in \sigma_1$, we have $X,(x,0)\vDash \Box_2 p$
because, for $x\not=0$, the set $\{x\}\times (-|x|,|x|)$ is an open neighborhood of $(x,0)$ in the horizontal topology $\sigma_2$ contained in $V(p)$ 
and for $x=0$, the set $\{0\}\times (-1,1)$ has these properties.
On the other hand $X,(0,0)\not\vDash \Box_2\Box_1 p$: Any open neighborhood in $\sigma_2$ of $(0,0)$ contains a point $(0,y)$ with $y\not=0$.
The point $(0,y)$ does not have a vertical neighborhood in $\sigma_1$ contained in $V(p)$.
\end{example}

\vspace{3pt}
\noindent\textbf{Full products.}
In the same work \cite{vanbenguram}, the authors also study the standard product topology  together with the horizontal and vertical topologies, yielding a \emph{full topological product} with three topologies. The modal logic of these products then uses $\Box_1$ and $\Box_2$ for the vertical and horizontal topologies and $\Box$ for the standard product topology.
 In these full products, the axiom $(\mathsf{sub})$ expressing that standard neighborhoods of  a state contain some vertical and horizontal neighborhoods as subsets
 is valid:
 \begin{align*}
		(\mathsf{sub}) & \:=\: \Box p \rightarrow \Box_1 p \land \Box_2 p.
\end{align*}
In Kripke semantics, the analogous full product has three accessibility relations: Namely, the horizontal and vertical accessibility relations $R_1$ and $R_2$ acting only in one of the two dimensions
and the \emph{synchronous} relation $R$ that corresponds to taking a step in both dimensions, i.e., $R=R_1\circ R_2 = R_2\circ R_1$.
This strong connection to $R_1$ and $R_2$ leads to the fact that $\Box  p \leftrightarrow \Box_1\Box_2 p$ is valid (and $\Box  p \leftrightarrow \Box_2\Box_1 p$ then follows by the validity of $(\mathsf{com})$) and hence $\Box$ can be seen as pure syntactic sugar in this setting.
Note also that, as shown in Example~\ref{ex:sub-failure} below,
$\mathsf{(sub)}$ is not valid in general on full products of Kripke frames. For reflexive Kripke frames, however, $(\mathsf{sub})$ is valid on their full products.

\begin{example}
\label{ex:sub-failure}
Let $F_1$ and $F_2$ be the Kripke frames depicted in Figure \ref{fig:subnotvalid} and consider their full product $F_1\times^+ F_2$ with  horizontal, vertical and synchronous transition relations depicted in Figure \ref{fig:subnotvalid}  as well.
Define $V(p)=\{(y,v)\}$.
Then $F_1 \times^+ F_2,(x,u) \vDash \Box p$, but
$F_1 \times^+ F_2,(x,u) \nvDash \Box_1 p$,
so $\mathsf{(sub)}$ fails.
 \end{example}

	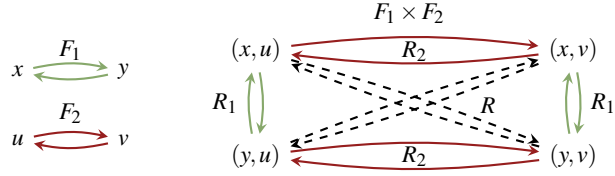
\begin{figure}
		\centering
	\begin{tikzpicture}[>=stealth, node distance=2cm, scale = 0.7]
		
		\node (x) at (-4,1.6) {\footnotesize$x$};
		\node (y) at (-2,1.6) {\footnotesize$y$};
		\draw[thick,->,sgreen] (x) to[bend left=12] (y);
		\draw[thick,->,sgreen] (y) to[bend left=12] (x);
		\node at (-3,2) {\footnotesize$F_1$};

		\node (u) at (-4,0.3) {\footnotesize$u$};
		\node (v) at (-2,0.3) {\footnotesize$v$};
        \draw[thick,->,wred] (u) to[bend left=12] (v);
		\draw[thick,->,wred] (v) to[bend left=12] (u);
		\node at (-3,0.8) {\footnotesize$F_2$};

		\node at (3.4,2.7) {\footnotesize$F_1 \times F_2$};
		
		\node (xu) at (0.5,2) {\footnotesize$(x,u)$};
		\node (xv) at (6.5,2) {\footnotesize$(x,v)$};
		\node (yu) at (0.5,0) {\footnotesize$(y,u)$};
		\node (yv) at (6.5,0) {\footnotesize$(y,v)$};
		
		\draw[thick,->,wred] (xu) to[bend left=7] node[below=-2.5pt, black] {\footnotesize $R_2$} (xv);
		\draw[thick,->,wred] (xv) to[bend left=7]  (xu);
		\draw[thick,->,wred] (yu) to[bend left=7]  (yv);
		\draw[thick,->,wred] (yv) to[bend left=7] node[above = -3pt, black] {\footnotesize$R_2$} (yu);
		
		\draw[thick,->,sgreen] (xu) to[bend left=10]  (yu);
		\draw[thick,->,sgreen] (yu) to[bend left=10] node[left, black] {\footnotesize$R_1$} (xu);
		\draw[thick,->,sgreen] (xv) to[bend left=10] node[right, black] {\footnotesize$R_1$} (yv);
		\draw[thick,->,sgreen] (yv) to[bend left=10]  (xv);
		
        \draw[dashed,->, thick]
        (xu) to[bend left=3]
        node[midway, above, yshift=-10pt, xshift=27pt] {\footnotesize $R$}
        (yv);
                \draw[dashed,->, thick] (yv) to[bend left=3]  (xu);

        \draw[dashed,->, thick]
        (xv) to[bend left=3](yu);
        \draw[dashed,->, thick] (yu) to[bend left=3]  (xv);
		
	\end{tikzpicture}
\caption{Frames $F_1$ and $F_2$ and their Kripke product.}
\label{fig:subnotvalid}

	\end{figure}

In this paper, we consider the analogous full products of neighborhood frames
equipped with vertical, horizontal, and the product neighborhood functions.
In contrast to the topological setting, $\mathsf{(sub)}$ can be falsified in this framework.
Indeed, Example~\ref{ex:sub-failure} transfers to neighborhood semantics:
passing to the associated neighborhood frames $N(F_1)$ and $ N(F_2)$ (see Section~\ref{par:validity})
preserves validity of modal formulas and  yields a counterexample
to $\mathsf{(sub)}$ for full products of neighborhood frames, too.

\vspace{3pt}
\noindent\textbf{Logics of products.}
We have already stated several axioms that are valid on different kinds of products of frames above. 
A naturally arising question is  what the precise logics of products of frames are in the different settings.
As a first ingredient, we need the \emph{fusion} of logics.
Given uni-modal logics $\Lambda_1$ and $\Lambda_2$, their fusion $\Lambda_1\otimes \Lambda_2$ is the smallest normal bi-modal logic containing
$\Lambda_1$ where $\Box$ has been renamed to $\Box_1$ and $\Lambda_2$ where $\Box$ has been renamed to $\Box_2$.
For three logics, $\Lambda_1$, $\Lambda_2$, and $\Lambda_3$, we keep the use of $\Box$ as it is in $\Lambda_3$, which will refer to the standard topology or product neighborhood function
in the analogously defined fusion $\Lambda_1\otimes \Lambda_2 \otimes \Lambda_3$.

Given modal logics $\Lambda_1$ and $\Lambda_2$, we denote by
$\Lambda_1 \times \Lambda_2$ the logic of the class of Kripke frames
$\{\, F_1 \times F_2 \mid F_1 \vDash \Lambda_1 \text{ and } F_2 \vDash \Lambda_2 \,\}$,
and by $\Lambda_1 \times^+ \Lambda_2$ the logic of the corresponding
full products of such frames.
Likewise, we write $\Lambda_1\times_t \Lambda_2$, $\Lambda_1\times_t^+ \Lambda_2$, $\Lambda_1\times_n \Lambda_2$, and 
$\Lambda_1\times_n^+ \Lambda_2$ for the logic of (full) products of topological spaces (t) or neighborhood frames (n) satisfying $\Lambda_1$ and $\Lambda_2$, respectively.

The first fundamental result on products in Kripke semantics, due to \cite{vbsproduct}, states that
\begin{align*}
\mathsf{K}\times \mathsf{K}
=
\mathsf{K}\otimes\mathsf{K}
+ (\mathsf{com})
+ (\mathsf{chr}).
\end{align*}
Moreover, for full products the additional modality is characterized by
$
\Box p \leftrightarrow \Box_1\Box_2 p$,
so the logic of a  full product of Kripke frames simply extends  $\mathsf{K}\otimes\mathsf{K}\otimes\mathsf{K}
+ (\mathsf{com})
+ (\mathsf{chr})$ with $\Box p \leftrightarrow \Box_1\Box_2p$, which {describes} $\Box$ in terms of $\Box_1$ and $\Box_2$.

As mentioned before, topological spaces always validate $\mathsf{S4}$.
The key results shown in \cite{vanbenguram} are as follows:
$
\mathsf{S4}\times_t \mathsf{S4}
=
\mathsf{S4}\otimes \mathsf{S4}$ and 
$
\mathsf{S4}\times_t^+ \mathsf{S4}
=
\mathsf{S4}\otimes \mathsf{S4}\otimes \mathsf{S4}
+ (\mathsf{sub}).
$
Since topological spaces form a special class of neighborhood frames, we have $\mathsf{S4}\times_n \mathsf{S4} = \mathsf{S4}\times_t \mathsf{S4}$ and $\mathsf{S4}\times_n^+ \mathsf{S4} = \mathsf{S4}\times_t^+ \mathsf{S4}$.
Thus,  
\begin{center}
$
\mathsf{S4}\times_n \mathsf{S4}
=
\mathsf{S4}\otimes \mathsf{S4},
\qquad
\mathsf{S4}\times_n^+ \mathsf{S4}
=
\mathsf{S4}\otimes \mathsf{S4}\otimes \mathsf{S4}
+ (\mathsf{sub}).$
\end{center}
The logic of full products in derivational semantics was  investigated for frames satisfying $\mathsf{D4}$.
The main results of \cite{KudDS} are
$
\mathsf{D4}\times_d \mathsf{D4}
=
\mathsf{D4}\otimes \mathsf{D4}$ and $
\mathsf{D4}\times_d^+ \mathsf{D4}
=
\mathsf{D4}\otimes \mathsf{D4}\otimes \mathsf{D4}
+ (\mathsf{sub})$.
These equalities extend analogously to neighborhood products:
\begin{center}$
\mathsf{D4}\times_n \mathsf{D4}
=
\mathsf{D4}\otimes \mathsf{D4},
\qquad
\mathsf{D4}\times_n^+ \mathsf{D4}
=
\mathsf{D4}\otimes \mathsf{D4}\otimes \mathsf{D4}
+ (\mathsf{sub}).$
\end{center}
Further, for
$L_1, L_2 \in \{\mathsf{S4}, \mathsf{D4}, \mathsf{D}, \mathsf{T}\}$,
it was shown in \cite{Kud2012} that
$
L_1 \times_n L_2
=
L_1 \otimes L_2$.
In this paper, we aim to extend this picture on neighborhood products of modal logics by determining the full neighborhood products 
$\mathsf{T}\times_n^+ \mathsf{T}$ and
$\mathsf{D}\times_n^+ \mathsf{D}$.

Finally, the following result generalizes some of the results stated before:
A normal logic $L$ is called an HTC-logic (Horn preTransitive Closed logic) if it is axiomatizable by closed formulas (i.e., formulas without propositional variables) together with formulas of the form
$\Box p \to \Box^{\,n} p$, $n \ge 0$.
For HTC-logics $L_1$ and $L_2$ one has 
$
L_1 \times_n L_2
=
L_1 \otimes L_2 + \Delta$,
where
$
\Delta =
\{\varphi \to \Box_2 \varphi \mid \varphi \text{ closed and } \Box_2\text{-free}\}
\cup
\{\varphi \to \Box_1 \varphi \mid \varphi \text{ closed and } \Box_1\text{-free}\}$ \cite{Kudhorn}.
If, moreover, $L_1$ and $L_2$ are serial, then 
$
L_1 \times_n L_2
=
L_1 \otimes L_2$ \cite{Kudhorn}.
The previously stated results for
$\mathsf{S4}, \mathsf{D4}, \mathsf{D}$, and $\mathsf{T}$
are special cases of this result.

\vspace{3pt}
\noindent\textbf{Interaction axioms $(\mathsf{sub})$ and $(\mathsf{mix})$.}
While the axiom $(\mathsf{sub})$ introduced above arises naturally in the context of full topological products, its validity crucially relies on the reflexivity of the underlying frames.
If one aims to move below the $\mathsf{S4}$-setting by moving to neighborhood frames,
a further structural interaction principle arises:
\[
(\mathsf{mix}) = \Box p \to \Box_1\Box_2 p \land \Box_2\Box_1 p.
\]
Over $\mathsf{S4}\otimes\mathsf{S4}\otimes\mathsf{S4}$, the axioms
$
(\mathsf{sub})\;=\; \Box p \to \Box_1 p \land \Box_2 p$ and 
$(\mathsf{mix})$ are equivalent as shown by the following syntactic derivation (detailed definitions used here are postponed until Section \ref{sec:prelim}):

\begin{proposition}
\label{prop:submix}
Over $\mathsf{S4}\otimes\mathsf{S4}\otimes\mathsf{S4}$, the axioms
$
(\mathsf{sub})\;=\; \Box p \to \Box_1 p \land \Box_2 p$ and 
$(\mathsf{mix})\:=\: \Box p \to \Box_1\Box_2 p \land \Box_2\Box_1 p$ are equivalent.
\end{proposition}

   \begin{proof}
\emph{$(\mathsf{sub}) \Rightarrow (\mathsf{mix})$.}
Assume $(\mathsf{sub})$, i.e.\ $\Box p \to \Box_1 p$ and
$\Box p \to \Box_2 p$.
From $\Box p \to \Box_1 p$, by $\Box_2$-Necessitation and $(K_{\Box_2})$, we obtain
$
\Box_2\Box p \to \Box_2\Box_1 p.
$
From $\Box p \to \Box_2 p$, by substitution $ \Box p$ for $p$, we get
$\Box\Box p \to \Box_2\Box p$.
Using $\mathsf{(4)}$, we derive
$\Box p \to \Box_2\Box p$,
and hence
$
\Box p \to \Box_2\Box_1 p.
$
The derivation of $\Box p \to \Box_1\Box_2 p$ is symmetric.

\medskip
\noindent
\emph{$(\mathsf{mix}) \Rightarrow (\mathsf{sub})$.}
From $(\mathsf{mix})$ we have $\Box p \to \Box_1\Box_2 p$.
By reflexivity $\Box_1 q \to q$, substituting $ \Box_2 p$ for $q$, we obtain
$\Box_1\Box_2 p \to \Box_2 p$,
and hence $\Box p \to \Box_2 p$.
The implication $\Box p \to \Box_1 p$ is obtained symmetrically.
 \end{proof}

So, in the context of topological semantics, both axioms may be used interchangeably. 
Moreover, the derivation of $(\mathsf{sub})$ from $(\mathsf{mix})$ only use $\mathsf{(T)}$ and so in fact $(\mathsf{sub})\in \mathsf{T}\otimes \mathsf{T}\otimes \mathsf{T} +  (\mathsf{mix})$, while the derivation of $(\mathsf{mix})$ from $(\mathsf{sub})$ only relies on $\mathsf{(4)}$ and so $(\mathsf{mix})\in \mathsf{K4}\otimes \mathsf{K4}\otimes \mathsf{K4} +  (\mathsf{sub})$. 

\paragraph*{Contributions.}
Using neighborhood semantics  to leave the realm of $\mathsf{S4}$-spaces, we  investigate full neighborhood products for the modal logics  $\mathsf{T} = \mathsf{K} + \Box p \to p$ and $\mathsf{D}=  \mathsf{K} + \Box p \to \Diamond p$, extending the
topological $\mathsf{S4}$-case to genuinely weaker systems. 
Our main result is:
\begin{restatable}[Main result]{theorem}{mainresult}
$
\mathsf{T}\times_n^+ \mathsf{T} = \mathsf{T}\otimes \mathsf{T}\otimes \mathsf{T}  + (\mathsf{mix})$ and $
 \mathsf{D}\times_n^+ \mathsf{D} =  \mathsf{D}\otimes \mathsf{D}\otimes \mathsf{D} + (\mathsf{mix}) 
$.
\end{restatable}
\noindent
Hidden in this presentation is the fact that   the reflexivity axiom $\Box p \to p$ on neighborhood frames leads to  full products satisfying $(\mathsf{sub})$ as $(\mathsf{sub})\in \mathsf{T}\otimes \mathsf{T}\otimes \mathsf{T} +  (\mathsf{mix})$. Since the frame from Example~\ref{ex:sub-failure} is serial, it
immediately yields a counterexample to $(\mathsf{sub})$ over
non-reflexive serial neighborhood frame products showing $(\mathsf{sub})\not \in \mathsf{D}\times_n^+ \mathsf{D}$.

\section{Preliminaries}
\label{sec:prelim}

We briefly introduce our notation, in particular for Kripke and neighborhood semantics of modal logic.

\vspace{3pt}
\noindent\textbf{Sequences.}
For a non-empty set $A$, we denote by $A^*$ the set of all finite sequences over $A$, including the empty
	sequence $\varepsilon$.
		For $\vec a = a_1 \cdots a_k \in A^*$, we write $|\vec a| = k$, with $|\varepsilon| = 0$.
	Concatenation  of sequences is denoted by $\cdot$.
	We identify sequences of length one with elements of $A$.

\vspace{3pt}
\noindent\textbf{Syntax of modal logic.}
Multimodal logic with $n$ primitive unary modal operators $(\Box_i)_{i\in \{1,\dots,n\}}$  is called $n$-modal logic.
If $n=1$, we call the logic uni-modal; if $n=2$, we call it bi-modal; if $n=3$, we call it tri-modal.
 Given a set of atomic propositions $\AP$,
the well-formed formulas of $n$-modal logic are given by the following grammar:
$\varphi::=  \bot \mid p\mid \neg \varphi \mid \varphi \land \varphi \mid \Box_i \varphi$
where $\bot$ is the Boolean constant false, $p\in \AP$, and $i\in \{1,\dots,n\}$. We use the usual Boolean abbreviations as well as 
the abbreviation $\Diamond_i \varphi$ for $\neg \Box_i \neg \varphi$. 

A normal $n$-modal logic is a set of formulas $\Lambda$ that contains all propositional tautologies and the $\mathsf{K}$-axiom
$\Box_i (p \to q) \to (\Box_i p \to \Box_i q)$ for all $1\leq i \leq n$ and that is closed under 
\begin{itemize}
\item modus ponens, i.e., if $\varphi\in \Lambda$ and $\varphi \to \psi\in \Lambda$, then $\psi\in \Lambda$,
\item generalization, i.e., if $\varphi \in \Lambda$, then $\Box_i\varphi \in \Lambda$ for $1\leq i \leq n$,
\item uniform substitution, i.e., if $\varphi\in \Lambda$, then also $\varphi[p\setminus \psi]\in \Lambda$ for all $p\in\AP$ and all formulas $\psi$
where $\varphi[p\setminus \psi]$ is obtained from $\varphi$ by replacing all occurrences of $p$ with $\psi$. 
\end{itemize}
We write $\mathsf{K}_n$ for the smallest normal $n$-modal logic over $\{\Box_1,\dots,\Box_n\}$.
Given a normal $n$-modal logic   $\Lambda$ and formulas $\varphi_1,\dots,\varphi_k$, we write 
$\Lambda + \varphi_1+\dots + \varphi_k$ to denote the smallest normal $n$-modal logic containing $\Lambda$ and  $\varphi_1,\dots,\varphi_k$.

\vspace{3pt}
\noindent\textbf{Fusion of logics.}	Let $L_1$ and $L_2$ be normal modal logics in the unimodal language with operator $\Box$.
	Write $L_i'$ for the result of uniformly replacing every occurrence of $\Box$ in formulas of $L_i$
	by $\Box_i$ (for $i\in\{1,2\}$).  
	The \emph{fusion} $L_1 \otimes L_2 $ of $L_1$ and $L_2$ is the smallest normal bimodal logic over $\{\Box_1,\Box_2\}$ containing $L_1'$ and $L_2'$. 
This definition generalizes straightforwardly to arbitrary finite $n$.

\vspace{3pt}
\noindent\textbf{Kripke semantics.}
Kripke semantics of $n$-modal logic are given in terms of Kripke frames and models.
A Kripke model is a tuple $M=(W,R_1,\dots,R_n,V)$ where $W$ is a non-empty set of \emph{worlds} or \emph{states} called the \emph{domain} or \emph{universe}, $R_i\subseteq W\times W$ is a binary 
\emph{accessibility relation} for each $i\leq n$, and $V\colon \AP \to 2^W$ is a \emph{valuation}. The tuple  $F=(W,R_1,\dots,R_n)$ is called a \emph{Kripke frame} and $M$ is said to be based on $F$.
The satisfaction relation $M,x\vDash\varphi$ is defined recursively by
$	M,x\vDash p \Leftrightarrow x\in V(p)$, 
	$M,x\vDash \Box_i\varphi \Leftrightarrow \forall y\,(xR_i y\Rightarrow M,y\vDash\varphi)$,
and the standard semantics of $\bot,\neg,\land$. We write $\Diamond_i\varphi:=\neg\Box_i\neg\varphi$.


\vspace{3pt}
\noindent\textbf{Products of Kripke frames.}
Given two uni-modal Kripke frames $F=(W,S)$ and $F'=(W',S')$ we define 
the relations $R_1$ and $R_2$ on $W\times W'$ as follows:
for all $w,w'\in W$ and $v,v'\in W'$, we define $(w,v) R_1 (w',v')$ 
iff $wSw'$ and $v=v'$ as well as $(w,v) R_2 (w',v')$ iff $vS'v'$ and $w=w'$. 
Further, we define $R=R_1\circ R_2$, which then satisfies $(w,v) R_1 (w',v')$  iff $wSw'$ and $vS'v'$ for all $w,w'\in W$ and $v,v'\in W'$.
The product of the frames is the bi-modal Kripke frame $F\times F' = (W\times W',R_1,R_2)$ and the 
full product is the frame $F\times^+ F'=(W\times W',R_1,R_2,R)$.

\vspace{3pt}
\noindent\textbf{Neighborhood semantics.}
Let $X$ be a non-empty set and $k \ge 1$.
A \emph{$k$-neighborhood frame} (\emph{$k$-n-frame}) is a pair $(X,\vec{\tau})$, where
$\vec{\tau}=(\tau_1,\dots,\tau_k)$ and each $\tau_i : X \to 2^{2^X}$ assigns to every
$x\in X$ a \emph{filter} $\tau_i(x)$ on $X$, i.e., a non-empty collection of subsets
of $X$ closed under supersets and finite intersections.
A \emph{$k$-neighborhood model} (\emph{$k$-n-model}) is a triple
$(X,\vec{\tau},\nu)$, where $\nu : \mathsf{AP} \to 2^X$ is a valuation.
If $k=1$ or  it is clear from context, we simply speak of neighborhood frames and models.
For $M=(X,\tau_1,\dots,\tau_k,\nu)$ and $x\in X$, satisfaction is defined by
$M,x\vDash p$ iff $x\in\nu(p)$ and
$M,x\vDash \Box_i\varphi$ iff $\llbracket\varphi\rrbracket^M\in\tau_i(x)$,
where $\llbracket\varphi\rrbracket^M=\{\,y\in X\mid M,y\vDash\varphi\,\}$;
the clauses for $\bot,\neg,\land$ are standard.


\vspace{3pt}
\noindent\textbf{Validity and logic of frames.}\label{par:validity}  
A formula $\varphi$ is valid on a Kripke or neighborhood frame $F$, written $F\vDash \varphi$, if $M,x\vDash \varphi$ for all models $M$ based on $F$ and all worlds $x$ of $F$.  For a set of formulas $\Gamma$, we write $F \vDash \Gamma$ if $F \vDash \varphi$ for every $\varphi \in \Gamma$. The logic of a class of neighborhood  or Kripke frames $\mathcal{C}$ is defined as $\Log(\mathcal{C}) = \{\varphi \mid F \vDash \varphi \text{ for all } F \in \mathcal{C}\}$. If $\mathcal{C}=\{F\}$ for some $F$, we write $\Log(F)$. We call a frame $F$ an $L$-frame if $\Log(F) \supseteq L$.
 For a set of formulas $\Gamma$, we define $\NF(\Gamma) = \{F \mid F$ is a neighborhood frame and $F \vDash \Gamma\}$ and $\KF(\Gamma) = \{F \mid F$ is a Kripke frame and $F \vDash \Gamma\}$.
    
Given a Kripke frame    $F = (W,R_1,\dots,R_k)$, we  let $R_i(w) =\{v\in W \mid wR_i v\} $ for  $w\in W$ and $1\leq i \leq k$ and define an $k$-n-frame $N(F) = (W, \tau_1,\dots,\tau_k)$  by setting    $\tau_i(w) = \{ U \mid R_i(w) \subseteq U \subseteq W \}$.
Then, it is well-known that $\Log(F) = \Log(N(F))$ \cite{Pacuit}.

\phantomsection\label{neighborhood-T-D}
\noindent\textbf{Neighborhood conditions for (T) and (D).}
We recall standard semantic characterizations of the axioms (T) and (D)
on neighborhood frames (see, e.g.,~\cite{Pacuit}). Let $(X, \tau)$
be a neighborhood frame. Then:
\begin{itemize}
	\item (T): $(X, \tau) \models \Box p \to p$ iff for all $\nu \in X$ and all $U \in \tau(\nu)$,
	we have $\nu \in U$.
	\item (D): $(X, \tau) \models \Box p \to \Diamond p$ iff for all $\nu \in X$,
	$\emptyset \notin \tau(\nu)$.
\end{itemize}

\vspace{3pt}
\noindent\textbf{Bounded morphism.}
Let $F=(W,R_1,\dots,R_n)$ and $F'=(W',R'_1,\dots,R'_n)$ be Kripke frames. A function $f\colon W\to W'$ is called a \emph{bounded morphism}
if
\begin{itemize}
\item for each $w,v\in W$ and $i\leq n$ with $w R_i  v$, we have $f(w) R_i' f(v)$,
\item for each $w\in W$,  $v'\in W'$  and $i\leq n$ with $ f(w) R_i' v'$, 
there is a $v\in W$ with $w R_i v$ and $f(v)=v'$. 
\end{itemize}

   	Let $F= (W, \tau_1,\dots, \tau_k)$ and $F' = (W', \sigma_1,\dots, \sigma_k)$ be neighborhood frames.
   	A function $f : W \to W'$ is called a bounded morphism if
   	\begin{itemize}
   		\item for all $1 \le i \le k$, for all $w \in W$ and all $U \in \tau_i(w)$, there is a $V\in \sigma_i(f(w))$ with $V = f(U)$,
   		\item for all $1 \le i \le k$, for all $w \in W$ and all $V \in \sigma_i(f(w))$, there exists $U \in \tau_i(w)$ such that $f(U) \subseteq V$.
   	\end{itemize}
    If there is a surjective bounded morphism between neighborhood or Kripke frames $F$ and $F'$, we call $F'$ a bounded morphic image of $F$ and write
$F \twoheadrightarrow  F'$.
     It is well-known that 
    	$\Log(\{F\}) \subseteq \Log(\{F'\})$ in this case.

\section{Lower bounds on logics of products of neighborhood frames}

    Let $\mathcal{X} = (X, \tau_1)$ and $\mathcal{Y} = (Y, \tau_2)$ be two unimodal neighborhood frames. Then the product of these two frames
    is a bimodal neighborhood frame  $ \mathcal{X} \times_n \mathcal{Y} = (\mbox{X} \times \mbox{Y}, \tau_1', \tau_2')$  defined as follows : 
    \begin{align*}
   \tau_1'(x,y) &= \{ U \subseteq \mbox{X} \times \mbox{Y} \mid \exists V \in \tau_1(x) : V \times  \{ y \} \subseteq U \}, \\
   \tau_2'(x,y) &= \{ U \subseteq \mbox{X} \times \mbox{Y} \mid \exists V \in \tau_2(y) : \{ x \} \times V \subseteq U \}.
    \end{align*}
    Here, $\tau_1'$ is called the \emph{horizontal neighborhood function} and $\tau_2'$ is called the \emph{vertical neighborhood function}.
    
    The full product of $\mathcal{X}$ and $\mathcal{Y}$ is the $3$-n-frame
    $\mathcal{X}\times_n^{+}\mathcal{Y}=(X\times Y,\tau'_1,\tau'_2,\tau)$,
    where $\tau'_1$ and $\tau'_2$ are the horizontal and vertical neighborhood functions, respectively, and $\tau$ is the \emph{product neighborhood function} given by
    \[
    \tau(x,y)=\{\,U\subseteq X\times Y \mid \exists W\in\tau_1(x)\,\exists V\in\tau_2(y)\ (W\times V\subseteq U)\,\}.
    \]

%
%
%
%

\begin{definition}
	For unimodal logics $L_1,L_2$, we define their neighborhood product and full neighborhood product:
	The neighborhood product is the bimodal logic
	\[
	L_1\times_n L_2 := \Log\{\mathcal{X}\times_n \mathcal{Y}\mid
	\mathcal{X}\in \NF(L_1),\ \mathcal{Y}\in \NF(L_2)\}.
	\]
	The full neighborhood product is the trimodal logic
	\[
	L_1\times_n^{+}L_2 = \Log\{\mathcal{X}\times^+_n \mathcal{Y}\mid
	\mathcal{X}\in \NF(L_1),\ \mathcal{Y}\in \NF(L_2)\}.
	\]
	We use $\Box_1$ and $\Box_2$ as modal operators for the horizontal and vertical neighborhood functions 
	and $\Box$ as the modal operator of the product neighborhood function.
	
\end{definition}

\paragraph*{Lower bounds on the full neighborhood products $\mathsf{T}\times_n^+ \mathsf{T}$ and $\mathsf{D}\times_n^+ \mathsf{D}$.}
As before, we use modal operators $\Box_1$, $\Box_2$, and $\Box$ for the horizontal, vertical, and standard neighborhood functions of full products of neighborhood frames.
Consequently, we also use $\Box_1$, $\Box_2$, and $\Box$ as the three modalities of a fusion $L_1\otimes L_2 \otimes L_3$.

In this paper, we study the full neighborhood products of neighborhood frames satisfying the $\mathsf{T}$-axiom $\Box p \to p$ and the $\mathsf{D}$-axiom $\Box p \to \Diamond p$.
	Over the modal language with modal operators
	$\Box_1$,  $\Box_2$, and $\Box$,
	we  need the axiom 	$(\mathsf{mix})  \:=\: \Box p \rightarrow \Box_1 \Box_2 p \land \Box_2 \Box_1 p$ as discussed in the introduction.
	Using $\Box_1$,  $\Box_2$, and $\Box$ in the fusions, we define the logics
	\begin{align*}
	\TNL  \:=\: \mathsf{T}\otimes \mathsf{T}\otimes \mathsf{T} + (\mathsf{mix}) \qquad \text{ and } \qquad
	\DNL  \:=\: \mathsf{D}\otimes \mathsf{D}\otimes \mathsf{D} + (\mathsf{mix}).
	\end{align*}
First, let us spell out what the logics express about Kripke frames:

\begin{restatable}{proposition}{Kripkemix}
For  Kripke frames $F=(W,R_1,R_2,R)$,  $F\vDash (\mathsf{mix})$ if and only if  $R_1\circ R_2,R_2\circ R_1\subseteq R$.
\end{restatable}

\begin{proof}
First assume  $F$ does not satisfy  $R_1\circ R_2,R_2\circ R_1\subseteq R$.
Then, w.l.o.g., there are  $x,y$ such that $x(R_1\circ R_2)y$ and $\neg xRy$. Choosing a valuation with $\nu(p):=R(x)$,
we get $F,\nu,x\vDash\Box p$ while  $F,\nu,x\nvDash \Box_1\Box_2 p$, showing that $(\mathsf{mix})$ is not valid on $F$.
Conversely, assume $F\nvDash (\mathsf{mix})$. So, w.l.o.g., there is a valuation $\nu$ and a world $x$ such that
$F,\nu,x\vDash \Box p$, but $F,\nu,x\nvDash \Box_1\Box_2 p$. So, there is a world $y$ with $x(R_1\circ R_2)y$, but  not $xRy$, showing that
$F$ does not satisfy $R_1\circ R_2,R_2\circ R_1\subseteq R$.	
\end{proof}

As  for any accessibility relation $R'$
on a Kripke frame, the formula $(\mathsf{T})$ is valid iff $R'$ is reflexive, and
$(\mathsf{D})$ is valid iff $R'$ is serial, this means the following:

\begin{corollary}\label{prop:first-order}
For  Kripke frames $F=(W,R_1,R_2,R)$, we have that 
\begin{enumerate}
\item $F\vDash \TNL$ iff $R_1$, $R_2$, and $R$ are reflexive and $R_1\circ R_2,R_2\circ R_1\subseteq R$ (which implies $R_1,R_2\subseteq R$).
\item $F\vDash \DNL$ iff $R_1$, $R_2$, and $R$ are serial and $R_1\circ R_2,R_2\circ R_1\subseteq R$.
\end{enumerate}
\end{corollary}

\noindent 
Now, we provide lower bounds for the logics of full neighborhood products (proof in the appendix).
\begin{restatable}{proposition}{lowerbound}
\label{prop:lowerbound}
We have
$\mathsf{T}\times_n^+ \mathsf{T} \supseteq \TNL$ and $\mathsf{D}\times_n^+ \mathsf{D} \supseteq \DNL $.
\end{restatable}

\begin{proof}
      To show $\mathsf{T}\times_n^+ \mathsf{T} \supseteq \TNL$, let $F_1$ and $F_2$ be
      $\mathsf{T}$-frames and consider their full product $F_1 \times_n^+ F_2$.
      Fix $(x,y) \in F_1 \times_n^+ F_2$. 
       We first verify the axiom $\Box p \to p$. Recall the standard fact \ref{neighborhood-T-D} that
       \[
       F,x \vDash \Box p \to p \quad\text{iff}\quad x \in U \text{ for all } U \in \tau(x).
       \]
       Now, assume $U \in \tau(x,y)$.
       By definition of the standard neighborhood function, there exist
       $V \in \tau_1(x)$ and $W \in \tau_2(y)$ with
       $V \times W \subseteq U$.
       Since $x \in V$ and $y \in W$, we obtain
       $(x,y) \in V \times W \subseteq U$. The axioms $\Box_1 p \to p$ and $\Box_2 p \to p$ are validated by analogous
       arguments. 
       
      We next verify the axiom $\Box p \to \Box_1 p \land \Box_2 p$.
      Assume that $F_1 \times_n^+ F_2,\nu,(x,y) \vDash \Box p$.
      By definition, there exists a set $U \subseteq \nu(p)$ such that
      $U \supseteq V \times W$ for some $V \in \tau_1(x)$ and $W \in \tau_2(y)$.
      Since $y \in W$, we have $V \times \{y\} \subseteq U$, and hence
      $U \in \tau_1'(x,y)$.
      The same argument applies to $\tau_2'(x,y)$.
      It follows that
      $F_1 \times_n^+ F_2,(x,y) \vDash \Box_1 p \land \Box_2 p$.

For $\Box p \to \Box_1 \Box_2 p$, assume $F_1 \times_n^+ F_2, \nu, (x,y) \vDash \Box p$. Then there exists $U \in \tau(x,y)$ such that $U \subseteq \nu(p)$ and $U_1 \times U_2 \subseteq U$ for some $U_1 \in \tau_1(x)$ and $U_2 \in \tau_2(y)$. We show $F_1 \times_n^+ F_2, \nu,  (x,y) \vDash \Box_1 \Box_2 p$. Since $U_1 \times \{y\} \in \tau_1'(x,y)$, it suffices to prove $U_1 \times \{y\} \subseteq \nu(\Box_2 p)$. Let $(v',y) \in U_1 \times \{y\}$. As $U_2 \in \tau_2(y)$, we have $\{v'\} \times U_2 \in \tau_2'(v',y)$. Moreover, $\{v'\} \times U_2 \subseteq U_1 \times U_2 \subseteq U \subseteq \nu(p)$. Hence $(v',y) \vDash \Box_2 p$. The argument for $\Box p \to \Box_2 \Box_1 p$ is symmetric.

For $\mathsf{D}\times_n^+ \mathsf{D} \supseteq \mathsf{D}_3^m$, let $F_1,F_2$ be $\mathsf{D}$-frames and fix $(x,y)$. To verify $\Box p \to \Diamond p$, it suffices to show $\emptyset \notin \tau(x,y)$ \ref{neighborhood-T-D}. Let $U \in \tau(x,y)$. Then $V \times W \subseteq U$ for some $V \in \tau_1(x)$ and $W \in \tau_2(y)$. Since $F_1,F_2 \vDash \mathsf{D}$, we have $V \neq \emptyset$ and $W \neq \emptyset$. Hence $U \neq \emptyset$. Thus $\emptyset \notin \tau(x,y)$. This completes the proof, as $\mathsf{(mix)}$ holds on all full products.
\end{proof}

\section{Upper bounds on the full neighborhood products \texorpdfstring{$\mathsf{T}\times_n^+ \mathsf{T}$}{T x\_n+ T} and \texorpdfstring{$\mathsf{D}\times_n^+ \mathsf{D}$}{D x\_n+ D}}

In this section, we now prove our main result by showing that $\mathsf{T}\times_n^+ \mathsf{T}\subseteq \TNL$ and that $\mathsf{D}\times_n^+ \mathsf{D}\subseteq \DNL$.
To this end, we first define two infinite trees $\Tsc$ and $\Tim$ with three relations and show soundness and completeness of $\TNL$ and $\DNL$ with respect to the respective trees (Section \ref{sec:trees}). Based on infinite trees with one relation for the logics $\mathsf{T}$ and $\mathsf{D}$,
we afterwards define neighborhood frames $\Nr$ and $\Ni$ validating $\mathsf{T}$ and $\mathsf{D}$, respectively, based on pseudo-infinite extensions of paths in the trees (Section \ref{sec:pseudoinfinite}).
The proof is then completed by providing a bounded morphism from $\Nr\times_n^+\Nr$ to $N(\Tsc)$ and from $\Ni\times_n^+\Ni$ to $N(\Tim)$ (Section \ref{sec:boundedmorphism}).

\subsection{Completeness of \texorpdfstring{$\TNL$}{TNL} and \texorpdfstring{$\DNL$}{DNL} with respect to single trees}
\label{sec:trees}

We prove that $\TNL$ is sound and complete with respect to a single reflexive infinite tree, and that $\DNL$ is sound and complete with respect to its irreflexive counterpart. 
To this end, we introduce the infinitely branching trees $\Tr$ and $\Ti$ of infinite depth, each equipped with a single relation, where $\Tr$ is reflexive and $\Ti$ is irreflexive. 
We show that $\mathsf{T}$ is complete for $\Tr$ and $\mathsf{D}$ for $\Ti$, results that will be used in the subsequent section.

Afterwards, we define an infinitely branching, infinite-depth reflexive tree $\Tsc$ with relations $R,R_1,R_2$ satisfying 
$R_1\circ R_2,\,R_2\circ R_1 \subseteq R$, thereby validating $\TNL$. 
On the irreflexive counterpart $\Tim$, the logic $\DNL$ is valid. 
We then prove soundness and completeness of $\TNL$ with respect to $\Tsc$ and of $\DNL$ with respect to $\Tim$. In \cite{vanbenguram}, a similar approach based on infinite binary trees with three relations is used to establish a completeness result for $\mathsf{S4}\otimes \mathsf{S4}\otimes \mathsf{S4}
+ (\mathsf{sub})$. 
Before establishing completeness, we recall the finite model property.

\begin{definition}[Finite model property]
	A logic $L$ has the \emph{finite model property} (FMP) if every formula
	$\varphi \notin L$ is falsified in some finite Kripke model of $L$;
	that is, there exist a finite model $M$ based on a frame $F$
	with $F \vDash L$ and a world $w$ such that $M,w \nvDash \varphi$.
\end{definition}

\begin{restatable}{proposition}{fmp}
	\label{prop:fmp}
	The logics $\TNL$ and $\DNL$ have the FMP.
\end{restatable}

    \begin{proof}
	We first prove the result for $\TNL$. The same approach works for $\DNL$ as we briefly explain at the end of the proof.
	
	We fix a class of Kripke frames $\mathcal C$ w.r.t.\ which $\TNL$ is sound and complete, and then obtain the FMP by filtration.
	Let $\mathcal C$ be the class of all frames $F=(W,R,R_1,R_2)$, with $F\vDash \TNL$.
	Soundness of $\TNL$ w.r.t.\ $\mathcal C$ is immediate, and since all axioms of $\TNL$ are Sahlqvist (\cite{Black}),
	the Sahlqvist theorem yields completeness of $\TNL$ w.r.t.\ $\mathcal C$.

	Let $\varphi \notin \TNL$. By completeness, there exist a $\TNL$-model
	$M=(W,R,R_1,R_2,V)$ and a world $w\in W$ such that $M,w \nvDash \varphi$.
	Let $\Gamma:=Sub(\varphi)$ be the set of subformulas of $\varphi$.
	A filtration through $\Gamma$ itself would not guarantee that the condition $ R_1\circ R_2 \;\cup\; R_2\circ R_1 \;\subseteq\; R$ is preserved.
	
	To circumvent this problem, we first enlarge $\Gamma$ to $\Gamma^{+}$ in the following way: for every $\Diamond\psi\in\Gamma$, we add
	$\Diamond_1\psi$ and $\Diamond_2\psi$. 
	Let now $M^{s}=(W_{\Gamma^+},R^{s},R_1^{s},R_2^{s},V_{\Gamma^{+}})$ be the minimal filtration of $M$ through $\Gamma^{+}$
	where $W_{\Gamma^+}=W/{\sim_{\Gamma^+}}$ and $\sim_{\Gamma^+}$ is the usual
	$\Gamma^+$-equivalence.
	To restore the required interaction, we redefine the $\Box$-relation by
	$$
	R^{s+}\;:=\;R^{s}\ \cup\ (R_1^{s}\circ R_2^{s})\ \cup\ (R_2^{s}\circ R_1^{s}),
	$$
	and write $M^{s+}= (F^{s+}, V_{\Gamma^+}) = (W_{\Gamma^+},R^{s+},R_1^{s},R_2^{s},V_{\Gamma^{+}})$.
	We show that $M^{s+}$ is still a filtration. Thus it remains to prove that for all
	$\Diamond\varphi\in\Gamma^{+}$ and all $[w],[v]\in W_{\Gamma^+}$,
	\[
	[w]\,R^{s+}\,[v]\ \text{ and }\ M,v\vDash \varphi
	\quad\Rightarrow\quad
	M,w\vDash \Diamond\varphi.
	\]
	If $[w]\,R^{s}\,[v]$, the claim holds as the minimal filtration $M^{s}$ is a filtration. Suppose $([w],[v])\in R_1^{s}\circ R_2^{s}$. Then there is $[u]\in W_{\Gamma^{+}}$ with
	$[w]R_1^{s}[u]$ and $[u]R_2^{s}[v]$. Choose representatives $w'\in[w]$, $u',u''\in[u]$,
	$v'\in[v]$ such that $w'R_1u'$ and $u''R_2v'$. Since $M,v\vDash\varphi$, we have
	$M,v'\vDash\varphi$, hence $M,u''\vDash\lozenge_2\varphi$ and therefore also
	$M,u'\vDash\lozenge_2\varphi$ (as $u'\sim_{\Gamma^{+}}u''$). It follows that
	$M,w'\vDash\lozenge_1\lozenge_2\varphi$. As $M$ validates $\TNL$, the axiom
	$\lozenge_1\lozenge_2\psi\rightarrow\lozenge\psi$ holds globally, so $M,w'\vDash\lozenge\varphi$.
	Finally, $w'\sim_{\Gamma^{+}}w$ yields $M,w\vDash\lozenge\varphi$.
	The case $([w],[v])\in R_2^{s}\circ R_1^{s}$ is symmetric.
	
	To complete the proof for $\TNL$, it remains to show that $F^{s+}\vDash\TNL$. 
	Reflexivity of all relations is preserved by any filtration. Further, $ R_1^{s}\circ R_2^{s}\subseteq R^{s+}$ and 
	$  R_2^{s}\circ R_1^{s}\subseteq R^{s+}$ 
	hold by the  definition of $R^{s+}$. 
	
	For $\DNL$, the same filtration works. As seriality is also preserved by any filtration, the constructed filtration  $M^{s+}$
	when starting with a $\DNL$-model $M$ is serial and satisfies $ R_1^{s}\circ R_2^{s}\subseteq R^{s+}$ and 
	$  R_2^{s}\circ R_1^{s}\subseteq R^{s+}$.
\end{proof}

\paragraph*{Trees $\Tr$, $\Ti$, $\Tsc$ and $\Tim$.}
Formally, $\Tr=(W,R)$ is the Kripke frame with $W=\mathbb N^*$, the set of all finite sequences over $\mathbb N$, and
$sRt$ iff there exists $u\in\mathbb N\cup\{\epsilon\}$ such that $t=s\cdot u$. 
The irreflexive variant is $\Ti=(W,R')$, where
$sR't$ iff there exists $u\in\mathbb N$ such that $t=s\cdot u$.

       Soundness and completeness of $\mathsf{T}$ and $\mathsf{D}$ with respect to $\Tr$ and $\Ti$, respectively, follow from the FMP of $\mathsf{T}$ and $\mathsf{D}$ together with a standard tree unfolding argument:
       
\begin{restatable}{proposition}{Tsoundandcomplete}
\label{prop:Tsoundcomplete}
    $\mathsf{T}$ is sound and complete w.r.t $\Tr$ and $\mathsf{D}$ is sound and complete w.r.t $\Ti$.
    \end{restatable}

    \begin{proof}
    Soundness is clear in both cases.
    For completeness of $\mathsf{T}$, we use the finite model property of $\mathsf{T}$, which yields
    $\mathsf{T} = \Log(\{F \mid F \text{ is finite and reflexive}\})$. Hence, to show $\Log(\Tr) \subseteq \mathsf{T}$, it suffices to construct a surjective bounded morphism from $\Tr$ onto an arbitrary finite reflexive frame.
    Now, let $F = (W', R')$ be such a finite rooted frame with root $w$. We recursively define an mapping $f$ from  the states of $\Tr$ to  states of $F$. Initially, we assign $w$ to the root of $\Tr$. Now suppose that a node $x \in \Tr$ is mapped to a node $u = f(x) \in F$. Let $s_0,\dots,s_{k-1}$ be the successors of $u$, where $k \geq 1$ is ensured by reflexivity. We map the successors $x\cdot 0, x\cdot 1, x\cdot 2, \dots$ of $x$ onto the successors of $u$ in a cyclic manner: for $i = k\cdot m + \ell$ with $m \in \mathbb{N}$ and $\ell < k$, we set $f(x\cdot i) = s_\ell$. By construction, $f$ is a bounded morphism.
    
    Completeness of $\mathsf{D}$ with respect to $\Ti$ is shown analogously.
    \end{proof}

We now turn to the tree $\Tsc$, which is also illustrated in Figure \ref{fig:treerm}. At each state, $R_1$ and $R_2$ give rise to disjoint countably infinite sets of successors. 
The relation $R$ strictly extends $R_1\cup R_2$ by adding a further countably infinite set of immediate successors and by satisfying 
$R_1\circ R_2 \subseteq R$ and $R_2\circ R_1 \subseteq R$. 
Formally, the domain of $\Tsc$ is 
$W=(\mathbb N\times\{0,1,2\})^*$, 
where the second component distinguishes the three types of successors described above. 
We define $\Tsc=(W,R,R_1,R_2)$ by stipulating, for all $s,t\in W$, that            
\begin{align*}
	s R_1 t 
	&\quad\text{iff}\quad
	t = s \cdot u \text{ for some } u \in \mathbb{N} \times \{1\} \cup \{\epsilon\},\\
	s R_2 t 
	&\quad\text{iff}\quad
	t = s \cdot u \text{ for some } u \in \mathbb{N} \times \{2\} \cup \{\epsilon\},\\
	s R t 
	&\quad\text{iff}\quad
	t = s \cdot u \text{ for some } u \in \mathbb{N} \times \{0\} \cup \{\epsilon\} \\
	&\qquad\quad\text{or there exists } q \in W \text{ such that } s R_1 q R_2 t
	\text{ or } s R_2 q R_1 t.
\end{align*}
Note that by the reflexivity of $R_1$ and $R_2$, the last condition also ensures $R_1,R_2\subseteq R$.

The tree $\Tim$ is defined analogously without making the relations reflexive. So, $\Tim = (W, R', R'_1, R'_2)$
where for all $s,t\in W$, we let 
$sR'_1 t$ iff $t=s\cdot u$ for some  $u \in \mathbb{N} \times \{1\} $;
$sR'_2 t$ iff $t=s\cdot u$ for some  $u \in \mathbb{N} \times \{2\} $;
$sR't$ iff $t=s\cdot u$ for some  $u \in \mathbb{N} \times \{0\} $ or there is a $q\in W$ with 
$s R'_1 q R'_2 t$ or  $s R'_2 q R'_1 t$.
Note that here $R'_1,R'_2\not \subseteq R'$.

    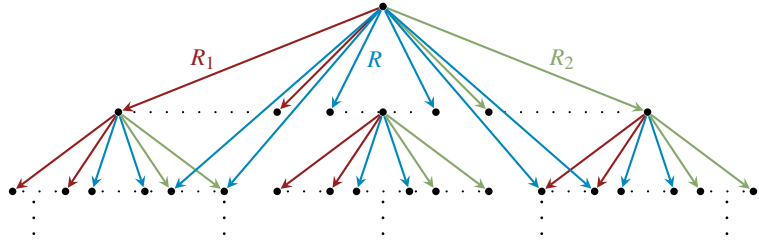
\begin{figure}[t]
	\centering
	\begin{tikzpicture}[scale=0.7]
		\node (K) at (0,0) [circle,fill=black,inner sep=1pt] {};
		\node (A) at (-5,-2) [circle,fill=black,inner sep=1pt] {};
		\node (B) at (-2,-2) [circle,fill=black,inner sep=1pt] {};
		\node (C) at (-1,-2) [circle,fill=black,inner sep=1pt] {};
		\node (K2) at (0,-2) [circle,fill=black,inner sep=1pt] {};
		\node (D) at (1,-2) [circle,fill=black,inner sep=1pt] {};
		\node (E) at (2,-2) [circle,fill=black,inner sep=1pt] {};
		\node (F) at (5,-2) [circle,fill=black,inner sep=1pt] {};
        \draw[line width=0.8pt, dash pattern=on 0pt off 6pt, line cap=round] (-4.7,-2) -- (-2.3,-2);
        \draw[line width=0.8pt, dash pattern=on 0pt off 6pt, line cap=round] (-.7,-2) -- (.7,-2);
        \draw[line width=0.8pt, dash pattern=on 0pt off 6pt, line cap=round] (2.3,-2) -- (4.7,-2);
        \draw[->, >=stealth, wred, thick](K) -- node[midway, left = 9pt] {\footnotesize $R_1$} (A);
        \draw[->, >=stealth, wred, thick] (K) -- (B);
        \draw[->, >=stealth, mblue, thick](K) -- node[midway, right = 0.1pt] {\footnotesize $R$} (C);
        \draw[->, >=stealth, mblue, thick] (K) -- (D);
        \draw[->, >=stealth, sgreen, thick] (K) -- (E);
        \draw[->, >=stealth, sgreen, thick](K) -- node[midway, right = 9pt] {\footnotesize $R_2$} (F);
        \node (A1) at (-7,-3.5) [circle,fill=black,inner sep=1pt] {};
        \node (B1) at (-6,-3.5) [circle,fill=black,inner sep=1pt] {};
        \node (C1) at (-5.5,-3.5) [circle,fill=black,inner sep=1pt] {};
        \node (D1) at (-4.5,-3.5) [circle,fill=black,inner sep=1pt] {};
        \node (E1) at (-4,-3.5) [circle,fill=black,inner sep=1pt] {};
        \node (F1) at (-3,-3.5) [circle,fill=black,inner sep=1pt] {};
		\draw[->, >=stealth, wred, thick] (A) -- (A1);       
		\draw[->, >=stealth, wred, thick] (A) -- (B1);
		\draw[->, >=stealth, mblue, thick] (A) -- (C1);
		\draw[->, >=stealth, mblue, thick] (A) -- (D1);
		\draw[->, >=stealth, sgreen, thick] (A) -- (E1);
		\draw[->, >=stealth, sgreen, thick] (A) -- (F1);
		\draw[line width=0.8pt, dash pattern=on 0pt off 6pt, line cap=round] (-6.8,-3.5) -- (-6.1,-3.5);
		\draw[line width=0.8pt, dash pattern=on 0pt off 6pt, line cap=round] (-5.3,-3.5) -- (-4.3,-3.5);
		\draw[line width=0.8pt, dash pattern=on 0pt off 6pt, line cap=round] (-3.8,-3.5) -- (-3.1,-3.5);
		\draw[->, >=stealth, mblue, thick] (K) -- (E1);
		\draw[->, >=stealth, mblue, thick] (K) -- (F1);
		\node (A2) at (-2,-3.5) [circle,fill=black,inner sep=1pt] {};
		\node (B2) at (-1,-3.5) [circle,fill=black,inner sep=1pt] {};
		\node (C2) at (-.5,-3.5) [circle,fill=black,inner sep=1pt] {};
		\node (D2) at (.5,-3.5) [circle,fill=black,inner sep=1pt] {};
		\node (E2) at (1,-3.5) [circle,fill=black,inner sep=1pt] {};
		\node (F2) at (2,-3.5) [circle,fill=black,inner sep=1pt] {};
		\draw[->, >=stealth, wred, thick] (K2) -- (A2);       
		\draw[->, >=stealth, wred, thick] (K2) -- (B2);
		\draw[->, >=stealth, mblue, thick] (K2) -- (C2);
		\draw[->, >=stealth, mblue, thick] (K2) -- (D2);
		\draw[->, >=stealth, sgreen, thick] (K2) -- (E2);
		\draw[->, >=stealth, sgreen, thick] (K2) -- (F2);
		\draw[line width=0.8pt, dash pattern=on 0pt off 6pt, line cap=round] (-1.8,-3.5) -- (-1.1,-3.5);
		\draw[line width=0.8pt, dash pattern=on 0pt off 6pt, line cap=round] (-.3,-3.5) -- (0.7,-3.5);
		\draw[line width=0.8pt, dash pattern=on 0pt off 6pt, line cap=round] (1.2,-3.5) -- (1.9,-3.5);
		\node (A3) at (3,-3.5) [circle,fill=black,inner sep=1pt] {};
		\node (B3) at (4,-3.5) [circle,fill=black,inner sep=1pt] {};
		\node (C3) at (4.5,-3.5) [circle,fill=black,inner sep=1pt] {};
		\node (D3) at (5.5,-3.5) [circle,fill=black,inner sep=1pt] {};
		\node (E3) at (6,-3.5) [circle,fill=black,inner sep=1pt] {};
		\node (F3) at (7,-3.5) [circle,fill=black,inner sep=1pt] {};
		\draw[->, >=stealth, wred, thick] (F) -- (A3);       
		\draw[->, >=stealth, wred, thick] (F) -- (B3);
		\draw[->, >=stealth, mblue, thick] (F) -- (C3);
		\draw[->, >=stealth, mblue, thick] (F) -- (D3);
		\draw[->, >=stealth, sgreen, thick] (F) -- (E3);
		\draw[->, >=stealth, sgreen, thick] (F) -- (F3);
		\draw[line width=0.8pt, dash pattern=on 0pt off 6pt, line cap=round] (3.2,-3.5) -- (3.9,-3.5);
		\draw[line width=0.8pt, dash pattern=on 0pt off 6pt, line cap=round] (4.7,-3.5) -- (5.6,-3.5);
		\draw[line width=0.8pt, dash pattern=on 0pt off 6pt, line cap=round] (6.2,-3.5) -- (6.9,-3.5);
		\draw[->, >=stealth, mblue, thick] (K) -- (A3);
		\draw[->, >=stealth, mblue, thick] (K) -- (B3);
		
		\foreach \x in {-6.6,-3,0, 3, 6.5}
		{
			\foreach \y in {-3.7, -4, -4.3}
			{
				\fill (\x,\y) circle (0.75pt);
			}
		}

	\end{tikzpicture}
\caption{
	The tree $\Tsc$ with countably infinite $R_1$-, $R_2$-, and $R$-successors,
	where $R_1, R_2 \subseteq R$ and 
	$R_1 \circ R_2,\, R_2 \circ R_1 \subseteq R$.
}
\label{fig:treerm}
\end{figure}

\begin{theorem}
\label{thm:Kripkecompleteness}
    $\TNL$ is sound and complete w.r.t $\Tsc$ and $\DNL$ is sound and complete w.r.t. $\Tim$.
    \end{theorem}

    \begin{proof}
    We start with the proof for $\TNL$.
Soundness is immediate: $\Tsc$ is reflexive with respect to all three relations, and the axiom $\mathsf{(mix)}$ is valid since $R_1\circ R_2 \subseteq R$ and $R_2\circ R_1 \subseteq R$.

For completeness, we use the finite model property of \(\TNL\).
Let \(F=(W',R'',R''_1,R''_2)\) be a finite rooted frame with root \(w\) such that \(F \vDash \TNL\).
We define a bounded morphism \(f\colon \Tsc\to F\) recursively, starting with \(f(\varepsilon)=w\).
Suppose \(f(x)=u\).
Then the \(R_1\)-successors \(x\cdot(m,1)\) and the \(R_2\)-successors \(x\cdot(m,2)\) \((m\in\mathbb N)\) are mapped cyclically onto the sets of \(R''_1\)- and \(R''_2\)-successors of \(u\), respectively. As our logics are serial, the sets of $R_1''$-, $R_2''$-, and $R''$-successors of $u$ are nonempty, so the construction can proceed.
For the \(R\)-successors of \(x\), we first assign only those of the form \(x\cdot(m,0)\), mapping them cyclically onto the \(R''\)-successors of \(u\).
The remaining \(R\)-successors, namely nodes of the form \(x\cdot(m,2)\cdot(m',1)\) and \(x\cdot(m,1)\cdot(m',2)\), are treated in the subsequent step of the construction.
    
%
%

The bounded morphism conditions for $R_1$ and $R_2$ are satisfied by construction, since the $R_i$-successors of any state $x$ are mapped onto the $R_i''$-successors of $f(x)$. It remains to verify the conditions for $R$. Let $x,y \in \Tsc$ with $xRy$.
If $y = x\cdot(m,i)$ for some $m\in\mathbb{N}$ and $i\in\{0,1,2\}$, then $f(y)$ is an $R''$-successor of $f(x)$: for $i=0$ this is immediate, and for $i=1,2$ it follows from $R_i'' \subseteq R''$, using $F \vDash \TNL$. Otherwise, $y$ is of the form $x\cdot(m,2)\cdot(m',1)$ or $x\cdot(m,1)\cdot(m',2)$ for some $m,m'\in\mathbb{N}$. Without loss of generality, assume $y = x\cdot(m,2)\cdot(m',1)$. Then $f(x) R_2'' f(x\cdot(m,2))$ and $f(x\cdot(m,2)) R_1'' f(y)$, and since $R_2'' \circ R_1'' \subseteq R''$, we obtain $f(x) R'' f(y)$.

For the back condition, let $x$ be a state of $\Tsc$ and $t\in W'$ with  $f(x)R''t$. Then, the step in the procedure mapping the $R$-successors of $x$ of the form $x\cdot (m,0)$ alternatingly to the $R''$-successors of $f(x)$ ensures that  $f(x\cdot (m',0))=t$ for some (in fact, infinitely many) $m'\in \mathbb{N}$. Surjectivity of $f$ is immediate from the construction. 

As $f$ is a surjective bounded morphism, it follows that
$\Log(\Tsc) \subseteq \Log(F)$. Since $F \vDash \TNL$ was arbitrary, we obtain $\Log(\Tsc) \subseteq \TNL$, which completes the proof of completeness.

The proof for $\DNL$ works analogously.
\end{proof}

\subsection{Neighborhood frame based on pseudo-infinite extensions of finite sequences}
\label{sec:pseudoinfinite}

To relate the Kripke completeness results from the previous section to neighborhood products, we make use of some auxiliary machinery described in the sequel.

\begin{remark}
	To obtain neighborhood frames $F_1$ and $F_2$ validating $\mathsf{T}$ with
	$\mathrm{Log}(F_1 \times^+_n F_2) = \mathsf{T}_3^m$, we cannot use the
	neighborhood versions of the Kripke frame $T^r_\omega$. Indeed, on full
	products of Kripke frames the interaction axioms (com) and (chr) are valid,
	and this validity is preserved when passing to the induced neighborhood
	frames $N(F)$. To falsify (com) and (chr), we therefore need
	neighborhood frames that genuinely do not arise from Kripke frames; a
	necessary condition for this is that some states do not have a smallest neighborhood.
	We thus introduce a neighborhood frame $N^r_\omega$ in which the neighborhoods
	of each state form an infinite strictly decreasing chain, so that  states do not
	possess a smallest neighborhood. The same construction applies to the
	irreflexive variant $N^i_\omega$, used in the $\mathsf{D}$-case.
\end{remark}

\paragraph*{Pseudo-infinite sequences.}
Based on the tree $\Tr$ whose states are finite sequences of natural numbers, 
we define a neighborhood frame $\Nr$ whose worlds are so-called \emph{pseudo-infinite sequences}. Constructions based on pseudo-infinite sequences are also used in
\cite{Kud2012,KudDS, Kudhorn}.
Formally, we define the set of pseudo-infinite sequences of natural numbers as 
\[
X = \{a_0a_1a_2 \ldots\in (\mathbb{N}\cup \{\star\})^\omega \mid \text{ there is an $N\in \mathbb{N}$ with $a_k=\star$ for all $k\geq N$}\}.
\]
We define the forget function $\forget\colon X\to\mathbb N^*$, which removes all occurrences of $\star$ from a sequence in $X$. 
We equip $X$ with a neighborhood function to obtain the neighborhood frame $\Nr$. 
Let $S$ be the successor relation of the tree $\Tr$, and for $\alpha=a_0a_1a_2\ldots\in X$ define:
 \begin{align*}
            st(\alpha) &= min\{N \mid \forall k \geq N : a_k = \star\} \\
            \alpha \mid_{k} &= a_0a_1...a_k \text{ for $k\in \mathbb{N}$}\\
            U_k(\alpha) &= \{ \beta\in X \mid \forget(\alpha)\,S\;\forget(\beta) \text{ in $\Tr$} \mbox{ and } \alpha \mid_m = \beta \mid_m,  \mbox{ where } m = max(k, st(\alpha))\} 
    \end{align*} 
These sets \(U_k(\alpha)\) define the neighborhoods of the point \(\alpha\in X\).
So, \(\Nr=(X,\tau)\), where the neighborhood function \(\tau\) is defined by \(\tau\colon X\to 2^{2^X}\) and
\[
\tau(\alpha)=\{\,U\subseteq X \mid \exists k\in\mathbb N\ (U_k(\alpha)\subseteq U)\,\}.
\]
Note that for any \(\alpha\in X\) with \(st(\alpha)=n\), we have \(U_n(\alpha)=U_j(\alpha)\) for all \(j\le n\).
Moreover, \(U_k(\alpha)\subseteq U_m(\alpha)\) whenever \(k\ge m\).
Indeed, let \(\beta\in U_k(\alpha)\). Since \(\alpha\mid_k=\beta\mid_k\) and \(k\ge m\), we have \(\alpha\mid_m=\beta\mid_m\), hence \(\beta\in U_m(\alpha)\).
Thus, for each \(\alpha\in X\), each \(U\in\tau(\alpha)\) contains some \(U_k(\alpha)\), and
\(U_0(\alpha)\supseteq U_1(\alpha)\supseteq\cdots\).
Starting at index \(st(\alpha)\), the sequence becomes strictly decreasing.
In particular, \(\mathsf T\) is valid on \(\Nr\): for every \(\alpha\in X\) and every \(U\in\tau(\alpha)\),
we have \(\alpha\in U\), since \(\alpha\in U_k(\alpha)\subseteq U\) for some \(k\).


Analogously, $\Ti$ induces a neighborhood frame $\Ni$ consisting of pseudo-infinite extensions of its nodes. 
As no state has the empty set as a neighborhood, $\mathsf D$ is valid on $\Ni$.
Analogously to \cite[Lemma 5.3]{Kud2012}, we obtain the following result.

\begin{lemma}
	The forget function defines a surjective bounded morphism
	from $\Nr$ to $N(\Tr)$ and from $\Ni$ to $N(\Ti)$.
	Hence $\Log(\Nr)=\mathsf{T}$ and $\Log(\Ni)=\mathsf{D}$.
\end{lemma}

\subsection{Surjective bounded morphism from \texorpdfstring{$\Nr\times_n^+\Nr$}{Nr x\_n+ Nr} to \texorpdfstring{$N(\Tsc)$}{N(Tsc)} and from \texorpdfstring{$\Ni\times_n^+\Ni$}{Ni x\_n+ Ni} to \texorpdfstring{$N(\Tim)$}{N(Tim)}}
\label{sec:boundedmorphism}

Now, we aim to show that $\mathsf{T}\times_n^+\mathsf{T}\subseteq \TNL$. We already know that $\TNL=\Log(\Tsc)$.
Further, we can view $\Tsc$ as a neighborhood frame $N(\Tsc)$ (as described in Section \ref{sec:prelim}) without changing the logic.
We also know that $\Nr\vDash \mathsf{T}$ and so $\mathsf{T}\times_n^+\mathsf{T}\subseteq \Log(\Nr\times_n^+\Nr)$.
Finally, we  show $\Log(\Nr\times_n^+\Nr)\subseteq \Log(N(\Tsc))$ by providing a surjective bounded morphism from $\Nr\times_n^+\Nr$ to $N(\Tsc)$.

   Using the domain $X$ and the neighborhood function of $\Nr$ as in the previous section, we write     
   \[
   \Nr \times^+_n \Nr = (X \times X, \tau_1, \tau_2, \tau) \qquad \text{ and } \qquad N(\Tsc) = (W, \sigma_1, \sigma_2, \sigma)
   \]
   where $W=(\mathbb{N}\times\{0,1,2\})^*$ and, for $x\in W$, the sets $\sigma_1(x)$, $\sigma_2(x)$, and $\sigma(x)$
  are the collections of all supersets of the sets of $R_1$-, $R_2$-, and $R$-successors of $x$ in $\Tsc$, respectively.
  
  To define the bounded morphism, we first fix a bijection $h\colon \mathbb{N}\times \mathbb{N}\to \mathbb{N}$.
    Next, we define function $f :(\mathbb{N} \cup \{\star \}) \times (\mathbb{N} \cup \{\star \}) \rightarrow \mathbb{N}\times\{0,1,2\} \cup \{\star\}$ as follows:
    \[
        f(a, b) =
        \begin{cases}
        \star, & \text{if }a=b=\star,\\
        (a,1), & \text{if } b = \star, a\in \mathbb{N}, \\
        (b,2), & \text{if } a = \star, b\in \mathbb{N}, \\
        (h(a, b),0), & \text{otherwise}.
        \end{cases}
    \]
Recall that $X$ is the domain of $\Nr$. We define a function
$g' \colon X \times X \to (\mathbb{N}\times\{0,1,2\}\cup\{\star\})^*\cdot\star^{\omega}$.
For $(\alpha,\beta) \in X \times X$ with $\alpha=a_0a_1\dots$ and $\beta=b_0b_1\dots$, we define
$
g'(\alpha,\beta)=f(a_0,b_0)f(a_1,b_1)\dots .
$
    
    Finally, we define $g(\alpha,\beta)$ as the sequence obtained from $g'(\alpha,\beta)$ by deleting all occurrences of $\star$. 
    For $\alpha=a_0a_1\dots$ and $\beta=b_0b_1\dots$ in $X$, there exists $N\in\mathbb N$ such that $a_k=b_k=\star$ for all $k\ge N$. 
    Hence $g'(\alpha,\beta)$ is eventually constant $\star$, and therefore 
    $g(\alpha,\beta)\in W=(\mathbb N\times\{0,1,2\})^*$.

        \begin{lemma}
        \label{lem:boundedmorphism}
        The function $g : \Nr \times^+_n \Nr \rightarrow N(\Tsc)$ is a surjective bounded morphism.
        \end{lemma}
        \begin{proof}
        First, we show that $g$ is surjective. Let $\vec z = z_0 z_1 \ldots z_n \in W=(\mathbb{N}\times\{0,1,2\})^*$.
        We define two sequences $\alpha=a_0 a_1 \ldots a_n \star^\omega$ and $\beta=b_0 b_1 \ldots b_n \star^\omega$
        as follows. For $i \le n$, we distinguish the following cases:
        \begin{itemize}
        \item If $z_i = (m,0)$ for some $m\in \mathbb{N}$, we pick $a_i,b_i\in \mathbb{N}$ such that $h(a_i,b_i)=m$.
        \item If $z_i = (m,1)$ for some $m\in \mathbb{N}$, we let $a_i = m$ and $b_i = \star$.
        \item If $z_i = (m,2)$ for some $m\in \mathbb{N}$, we let $b_i = m$ and $a_i = \star$.
        \end{itemize}
	This choice ensures that $g(\alpha,\beta) = \vec{z}$. Hence, g is surjective.

     We now verify the bounded morphism conditions for the neighborhood functions $\tau$ and $\sigma$. 
     The cases of $\tau_i$ and $\sigma_i$ for $i\in\{1,2\}$ are analogous and simpler, as fewer case distinctions are required. 
     To verify both conditions, we use the following claim:
        
                \vspace{6pt}
        \noindent\textbf{Claim: }
        For all $(\alpha, \beta)\in X\times X$ and  $m > \max\{st(\alpha), st(\beta)\}$,
        $R(g(\alpha,\beta)) = g(U_m(\alpha) \times U_m(\beta))$.
        \vspace{6pt}

        $\supseteq$ : Let $m > \max\{st(\alpha),st(\beta)\}$ and let $\zeta \in U_m(\alpha), \upsilon \in U_m(\beta)$.
        We show that $g(\zeta, \upsilon) \in R(g(\alpha,\beta))$. By definition, we need to find a $\vec{c} \in \mathbb{N}\times\{0,1,2\} \cup \{\epsilon\}$
        or $\vec{c} \in (\mathbb{N}\times\{1\}) \cdot (\mathbb{N}\times\{2\})$ or $\vec{c} \in (\mathbb{N}\times\{2\}) \cdot (\mathbb{N}\times\{1\})$  s.t. $g(\alpha,\beta) \cdot \vec{c} = g(\zeta, \upsilon)$.
        By the choice of $m>\max\{st(\alpha),st(\beta)\}$, we have
        $\alpha\mid_m=\zeta\mid_m$ and $\beta\mid_m=\upsilon\mid_m$,
        whence $g(\alpha,\beta)=g(\zeta\mid_m\star^\omega,\upsilon\mid_m\star^\omega)$.
        Furthermore, recalling that \(S\) denotes the successor relation on \(\Tr\), we have
        $
        \forget(\alpha)\;S\;\forget(\zeta)$ and $ \forget(\beta)\;S\;\forget(\upsilon)$.
        That means there exist $\vec d$ and $\vec e$, each of which is either empty ($\epsilon$) or a single letter from $\mathbb N$, such that
        $\forget(\alpha)\cdot \vec d = \forget(\zeta)$ and
        $\forget(\beta)\cdot \vec e = \forget(\upsilon)$. 
       Thus, the extensions $\vec d$ and $\vec e$ occur strictly beyond the prefixes $\zeta\mid_m$ and $\upsilon\mid_m$, respectively.
        Let $\zeta'$ and $\upsilon'$ denote the words after $\zeta \mid_m$ and $\upsilon \mid_m$, i.e., they are of the form: $\zeta' = \star^*\vec{d}\star^{\omega}$ and $\upsilon' = \star^* \vec{e}\star^{\omega}$. Now there are six cases :  
        \[
        \begin{cases}
        	g(\zeta', \upsilon') \in \mathbb{N}\times\{1\}, & \text{if } \vec{d} \neq \varepsilon \text{ and } \vec{e} = \varepsilon, \\
        	g(\zeta', \upsilon') \in \mathbb{N}\times\{2\}, & \text{if } \vec{d} = \varepsilon \text{ and } \vec{e} \neq \varepsilon, \\
        	g(\zeta', \upsilon') = \varepsilon, & \text{if } \vec{d} = \vec{e} = \varepsilon, \\
        	g(\zeta', \upsilon') \in \mathbb{N}\times\{0\}, & \text{if } \vec{d}, \vec{e} \neq \varepsilon \text{ occur at the same index in } \zeta', \upsilon', \\
        	g(\zeta', \upsilon') \in (\mathbb{N}\times\{1\}) \cdot (\mathbb{N}\times\{2\}), & \text{if } \vec{d}, \vec{e} \neq \varepsilon \text{ and } \vec{d} \text{ occurs in } \zeta' \text{ before } \vec{e} \text{ in } \upsilon', \\
        	g(\zeta', \upsilon') \in (\mathbb{N}\times\{2\}) \cdot (\mathbb{N}\times\{1\}), & \text{if } \vec{d}, \vec{e} \neq \varepsilon \text{ and } \vec{d} \text{ occurs in } \zeta' \text{ after } \vec{e} \text{ in } \upsilon'.
        \end{cases}
        \]

            It follows that $g(\zeta,\upsilon) = g(\zeta \mid_m \star^\omega, \upsilon \mid_m \star^\omega) \cdot g(\zeta', \upsilon') = g(\alpha, \beta) \cdot g(\zeta', \upsilon')$. Hence, $g(\zeta, \upsilon) \in R(g(\alpha,\beta))$. 
            \vspace{1pt}

            $\subseteq$ :  Let $\vec{w} \in R(g(\alpha,\beta))$ and  $m> max\{st(\alpha),st(\beta)\}$. We will construct $\zeta \in U_m(\alpha)$ and $\upsilon \in U_m(\beta)$ s.t. $g(\zeta,\upsilon) = \vec{w}$.
            By definition there exists $\vec{c} \in \mathbb{N}\times\{0,1,2\}\cup \{\epsilon\}$ or $\vec{c} \in (\mathbb{N}\times\{1\}) \cdot (\mathbb{N}\times\{2\})$ 
            or $\vec{c} \in (\mathbb{N}\times\{2\}) \cdot (\mathbb{N}\times\{1\})$  such that $g(\alpha,\beta) \cdot \vec{c} = \vec{w}$.
            We build $\zeta = \alpha \mid_m \cdot \vec{d} \, \star^\omega$ and $\upsilon = \beta \mid_m \cdot \vec{e} \, \star^\omega$ where $\vec{d}$ and $\vec{e}$ depend on the following cases:
        \[
            \begin{cases}
            \vec{d} = \vec{c},\; \vec{e} = \star, & \text{if }  \vec{c }\in \mathbb{N}\times\{1\} \\
            \vec{d} = \star,\; \vec{e} = \vec{c}, & \text{if }  \vec{c }\in \mathbb{N}\times\{2\} \\
            \vec{d} = \vec{e} = \star,& \text{if } \vec{c } = \epsilon \\
            \vec{d} = a,\; \vec{e} = b, & \text{if } \vec{c } \in \mathbb{N}\times \{0\} \text{ and } h(a,b) = \vec{c} \\
            \vec{d} = a,\; \vec{e} = \star b, & \text{if } \vec{c } \in (\mathbb{N}\times\{1\}) \cdot (\mathbb{N}\times\{2\}) \text{ with } \vec{c} = a \cdot b \\
            \vec{d} = \star a,\; \vec{e} = b, & \text{if } \vec{c } \in (\mathbb{N}\times\{2\}) \cdot (\mathbb{N}\times\{1\}) \text{ with } \vec{c} = b \cdot a 
            \end{cases}
        \] 
        We have $\zeta \in U_m(\alpha)$ and $\upsilon \in U_m(\beta)$. 
        Further,  $g(\zeta, \upsilon)=g(\zeta\mid_m, \upsilon\mid_m) \cdot g(\vec{d}\star^\omega,\vec{e}\star^\omega) = g(\alpha,\beta) \cdot \vec{c} = \vec{w}$.
        
        To verify the bounded morphism condition, let $\alpha,\beta\in X$ and $U\in\tau(\alpha,\beta)$.
        We need to find a $V\in\sigma(g(\alpha,\beta))$ with $V=g(U)$.
        Choose $m>\max\{st(\alpha),st(\beta)\}$ with
        $U_m(\alpha)\times U_m(\beta)\subseteq U$.
        Then, $        
        R(g(\alpha,\beta))
        = g\big(U_m(\alpha)\times U_m(\beta)\big)
        \subseteq g(U)$.
        
        For the second condition, let $\alpha,\beta\in X$ and
        $V\in\sigma(g(\alpha,\beta))$, that is,
        $R(g(\alpha,\beta))\subseteq V$.
        We need to find $U\in\tau(\alpha,\beta)$ with
        $g(U)\subseteq V$.
        Choose $m>\max\{st(\alpha),st(\beta)\}$ and set
        $U:=U_m(\alpha)\times U_m(\beta)$.
        Then, $
        g(U)
        = g\big(U_m(\alpha)\times U_m(\beta)\big)
        = R(g(\alpha,\beta))
        \subseteq V$.
        
     The arguments for $\tau_1,\sigma_1$ and $\tau_2,\sigma_2$ are analogous and simpler, and are therefore omitted.
        \end{proof}

Now, we are in the position to complete the proof of our first main result.

\begin{theorem}
	
We have $\mathsf{T}\times_n^+ \mathsf{T} = \TNL$.

\end{theorem}

\begin{proof}
The inclusion $\mathsf{T}\times_n^+\mathsf{T}\supseteq\TNL$ was established in Proposition~\ref{prop:lowerbound}. For the converse inclusion, since $\mathsf{T}\subseteq\Log(\Nr)$, we obtain $\mathsf{T}\times_n^+\mathsf{T}\subseteq\Log(\Nr\times_n^+\Nr)$. By Lemma~\ref{lem:boundedmorphism}, $\Log(\Nr\times_n^+\Nr)\subseteq\Log(N(\Tsc))$, and as observed earlier, passing from $\Tsc$ to $N(\Tsc)$ does not change the logic. Hence $\mathsf{T}\times_n^+\mathsf{T}\subseteq\Log(\Tsc)=\TNL$ by Theorem~\ref{thm:Kripkecompleteness}.
\end{proof}

To establish the analogous result for $\mathsf{D}\times_n^+ \mathsf{D}$, we reason analogously:

        \begin{lemma}
            $g : \Ni \times^+_n \Ni \rightarrow N(\Tim)$ is a bounded morphism where $g$ is as defined before.
                    \end{lemma}

            \begin{proof}
                Let $\alpha = a_0a_1\dots a_n\star^\omega$ and $\beta = b_0b_1\dots b_l\star^\omega$.
                Surjectivity follows as in the previous lemma. \newline
                We claim that
                $
                R(g(\alpha,\beta)) = g(U_m(\alpha) \times U_m(\beta))
                $
                for some $m > \max\{st(\alpha), st(\beta)\}$.
                This identity is used to verify the forth and back conditions for $\tau$ and $\sigma$.
                Fix such an $m > \max\{st(\alpha), st(\beta)\}$.
                \newline
                $\supseteq$: Assume $\zeta \in U_m(\alpha)$ and $\upsilon \in U_m(\beta)$. To show that $g(\zeta,\upsilon) \in R(g(\alpha,\beta))$, it suffices to find 
                $\vec{c} \in (\mathbb{N}\times\{0\}) \;\cup\; 
                ((\mathbb{N}\times\{1\})\cdot(\mathbb{N}\times\{2\})) \;\cup\;
                ((\mathbb{N}\times\{2\})\cdot(\mathbb{N}\times\{1\}))$, s.t.
                $g(\alpha,\beta) \cdot \vec{c} = g(\zeta,\upsilon)$. Since $m>\max\{st(\alpha),st(\beta)\}$, we have 
                $g(\alpha,\beta)=g(\zeta\mid_m,\upsilon\mid_m)$. Moreover, $\forget(\alpha) \,S\, \forget(\zeta)$ and $\forget(\beta) \,S\, \forget(\upsilon)$,
                where $S$ is the successor relation on $\Ti$.
                Thus there exist ${d},
                {e} \in \mathbb{N}$ such that
                $\forget(\alpha)\cdot{d}=\forget(\zeta)$ and
                $\forget(\beta)\cdot{e}=\forget(\upsilon)$.
                Hence ${d}$ and ${e}$ occur strictly after
                $\zeta\mid_m$ and $\upsilon\mid_m$.
                We denote $\zeta'$ and $\upsilon'$ to be the words after $\zeta\mid_m$ and $\upsilon\mid_m$ i.e. $\zeta' = \star^* {d} \star^\omega$ and $\upsilon' = \star^* {e} \star^\omega$.
                Then we have three cases :

         \[
            \begin{cases}
            g(\zeta', \upsilon') \in \mathbb{N}\times\{0\}, & {d} \text{ and } {e} \text{ are aligned in $\zeta'$ and $\upsilon'$} \\
            g(\zeta', \upsilon') \in (\mathbb{N}\times\{1\})\cdot(\mathbb{N}\times\{2\}), &  {d} \text{ has a lower word position  in $\zeta'$ than } {e} \text{in $\upsilon'$}\\
            g(\zeta', \upsilon') \in (\mathbb{N}\times\{2\})\cdot(\mathbb{N}\times\{1\}), &  {d}  \text{ has a higher word position  in $\zeta'$ than } {e} \text{in $\upsilon'$}
            \end{cases}
        \] 
        That means $g(\zeta,\upsilon)\in R(g(\alpha,\beta))$.\newline
        $\subseteq$: Let $\vec{w} \in R(g(\alpha,\beta))$. We construct $\zeta \in U_m(\alpha)$ and $\upsilon \in U_m(\beta)$ such that $g(\zeta,\upsilon)=\vec{w}$. By definition of $R$, there exists 
        $\vec{c} \in (\mathbb{N}\times\{0\}) \cup ((\mathbb{N}\times\{1\})\cdot(\mathbb{N}\times\{2\})) \cup ((\mathbb{N}\times\{2\})\cdot(\mathbb{N}\times\{1\}))$
        such that $g(\alpha,\beta)\cdot\vec{c}=\vec{w}$. Let $\zeta=\alpha\mid_m\cdot \vec{d}\star^\omega$ and $\upsilon=\beta\mid_m\cdot \vec{e}\star^\omega$, where $\vec{d}$ and $\vec{e}$ are chosen according to the following cases:
          \[
            \begin{cases}
            \vec{d} = a, \vec{e} = b & \text{if } \vec{c } \in \mathbb{N}\times\{0\} \text{ and } h(a, b) = \vec{c} \\
            \vec{d} = a, \vec{e} = \star b & \text{if } \vec{c } \in (\mathbb{N}\times\{1\}) \cdot (\mathbb{N}\times\{2\}) \text{ with } \vec{c} = ab \\
            \vec{d} = \star a, \vec{e} = b & \text{if } \vec{c } \in (\mathbb{N}\times\{2\}) \cdot (\mathbb{N}\times\{1\}) \text{ with } \vec{c} = b \cdot a 
            \end{cases}
        \] 
        Thus $\zeta \in U_m(\alpha)$ and $\upsilon \in U_m(\beta)$, and
        $g(\zeta,\upsilon)=g(\zeta\mid_m\star^\omega,\upsilon\mid_m\star^\omega)\cdot g(\vec d\star^\omega,\vec e\star^\omega)
        = g(\alpha,\beta)\cdot \vec c=\vec w$.
            \end{proof}

With the reasoning as above, this allows us to conclude the second main result:
\begin{theorem}
	
We have $\mathsf{D}\times_n^+ \mathsf{D} = \DNL$.

\end{theorem}

\section{Conclusion}
We identified the tri-modal logics for full products of neighborhood frames:
$\mathsf{T}\times_n^+ \mathsf{T} = \TNL$ and
$\mathsf{D}\times_n^+ \mathsf{D} = \DNL$.
We also showed that these logics possess the finite model property  and are therefore decidable.
A natural direction for future work is to study the products
$\mathsf{L_1}\times_n^+\mathsf{L_2}$ for
$L_1,L_2 \in \{\mathsf{K},\mathsf{K4}\}$.
Since these logics are not serial, the techniques used in this paper do not directly apply.

\bibliographystyle{plainurl}    
\bibliography{references}

%
%
%

\end{document}